\documentclass[oribibl]{llncs}
\usepackage{graphicx}
\usepackage{float}
\usepackage{hyperref}
\usepackage{algorithmic}
\usepackage{algorithm}
\usepackage{amssymb}
 \usepackage{amsmath}
\usepackage[listofformat=parens]{subfig}
\usepackage{tabularx}
\usepackage[T1]{fontenc}

\usepackage{tikz}
\usepackage{gnuplot-lua-tikz}
\tikzstyle{blank}=[rectangle,draw=white,fill=white,thick, inner sep=2pt,minimum
size=6mm]
\tikzstyle{cir}=[circle,draw=black,fill=white,thick, inner sep=1pt,minimum
size=3mm]
\tikzstyle{rcir}=[circle,draw=red,fill=white,thick, inner sep=1pt,minimum
size=3mm]

\def\er{Erd\"{o}s-R\'{e}nyi }

\title{On the Analysis of a Label Propagation Algorithm for
Community Detection\thanks{This work was done when the first author (KK) was
visiting The University of Iowa on an Indo-US Science and Technology Forum
Fellowship. The work of the second author (SP) was partially supported by
National Science Foundation grant CCF 0915543.}}
\author{Kishore Kothapalli\inst{1}  \and Sriram V. Pemmaraju\inst{2} \and Vivek
Sardeshmukh\inst{2}}
\institute{ International Institute of Information Technology, Hyderabad, India
500
032. \email{kkishore@iiit.ac.in} \and Department of Computer Science, The
University of Iowa, Iowa City, IA
52242-1419, USA. \email{firstname-lastname@uiowa.edu}}

\begin{document}
\maketitle 

\begin{abstract}
 This paper initiates formal analysis of a simple, distributed algorithm
for community detection on networks. We analyze an algorithm that we call
\textsc{Max-LPA}, both in terms of its convergence time and in terms of the
``quality'' of the communities detected. \textsc{Max-LPA} is an instance
of a class of community detection algorithms called \textit{label propagation}
algorithms. As far as we know, most analysis of label propagation algorithms
thus far has been empirical in nature and in this paper we seek a theoretical
understanding of label propagation algorithms.
In our main result, we define a clustered version of \er random graphs
with clusters $V_1, V_2, \ldots, V_k$ where the probability $p$, of an
edge connecting nodes within a cluster $V_i$ is higher than $p'$, the
probability of an edge connecting nodes in distinct clusters. We show
that even with fairly general restrictions on $p$ and $p'$ ($p =
\Omega\left(\frac{1}{n^{1/4-\epsilon}}\right)$ for any $\epsilon > 0$, $p' = O(p^2)$, where $n$ is
the number of nodes), \textsc{Max-LPA}
detects the clusters $V_1, V_2, \ldots, V_n$ in just two rounds. Based on this
and on empirical results, we conjecture that \textsc{Max-LPA} can correctly
and quickly identify communities on clustered \er graphs even when the
clusters are much sparser, i.e., with $p = \frac{c\log n}{n}$ for some $c > 1$.
\end{abstract}

\section{Introduction}
The problem of efficiently analyzing large social networks spans several areas
in computer science.
One of the key properties of social networks is their \textit{community
structure}. A \textit{community} in a network is a group of nodes that are
``similar'' to each other and ``dissimilar'' from the rest of the network.
There has been a lot of work recently on defining,
detecting, and identifying communities in real-world
networks~\cite{girvan2002community, flake2000efficient, raghavan2007near}.
It is usually, but not always, the tendency
for vertices to be gathered into distinct groups, or communities, such
that edges between vertices in the same community are dense but inter-community
edges are sparse~\cite{newman2003structure, girvan2002community}. A
\textit{community detection} algorithm takes as input a network and outputs a
partition of the vertex set into ``communities''. 
Detecting communities 
can allow us to understand attributes of vertices from the network topology
alone. 

There are many metrics to measure the ``quality'' of the 
communities detected by a community detection algorithm. A popular 
and widely adopted metric is 
\textit{graph modularity} defined by Newman~\cite{newman2004finding}. This
measure is obtained by summing up, over
all communities of a given partition, the difference between the observed
fraction of links inside the community and the expected value of this
quantity for a null
model, that is, a random network having the same size and same degree
sequence. Other popular measures include
\textit{graph conductance}~\cite{kannan2000clusterings} and \textit{edge
betweenness}~\cite{girvan2002community}. 

The community detection problem has connections to the \textit{graph
partitioning}
problem
which has been well studied since
1970s~\cite{elsner97graphpartitioning,karger2000minimum,kernighan1970efficient,suaris1988circuits}.
Graph partitioning problems are usually modeled as combinatorial optimization 
problems and this approach requires a precise sense of the objective function
being optimized.
Sometimes additional criteria such as the number of parts or
the sizes of parts also need to be specified.
In contrast, the notion of communities is relatively ``fuzzy'' 
\cite{fortunato2010community} and 
changes from application to application.
Furthermore, researchers in social network analysis are reluctant to
over-specify properties of communities and would rather let algorithms
``discover'' communities in the given network.
For a survey of the different approaches that have been proposed to 
find community structure in networks, see Fortunato's work 
\cite{fortunato2010community}. 

The focus of this paper is a class of seemingly simple community detection
algorithms called \textit{label propagation} algorithms (LPA).
Raghavan et al.~\cite{raghavan2007near} seem to be the first to study label
propagation algorithms for detecting network communities. 
The advantage of a LPA is, in addition to its simplicity, the fact that 
it can be easily parallelized or distributed. The generic LPA works as follows:
initially each node in the network is assigned a unique label.
In each iteration every node updates its label to the label which is the most
frequent in its neighborhood; ties are broken randomly.
One obtains variants of LPA by varying how the initial label assignment
is made, how ties are broken, and whether a node includes itself in
computing the most frequent label in its neighborhood.
In this paper, we analyze a specific instance of LPA called
\textsc{Max-LPA} in which nodes are assigned initial labels uniformly
at random from some large enough space.
Also, if there is a tie, it is broken in favor of the larger label.
Finally, a node includes its own label in determining the most
frequent label in its neighborhood.

At any point during the execution of a LPA, a community is simply all nodes 
with the same label. 
The intuition behind using a LPA for community detection is that 
a single label (the maximum label in the case of \textsc{Max-LPA}) can quickly 
become the most frequent label in neighborhoods within a dense cluster 
whereas labels have trouble ``traveling'' across a sparse set of edges
that might connect two dense clusters.
A LPA is said to have \textit{converged} if it starts
cycling through a collection of states.
Ideally, we would like LPA to converge to a cycle of period one, i.e., 
to a state in which any further execution of LPA yields the same state.
However, this is not always possible.
In fact, part of the difficulty of
analyzing LPA stems from the randomized tie-breaking rule. 
This way of breaking ties makes it difficult to estimate the period of the
cycle that the algorithm eventually converges to.
The version of LPA that we analyze, namely \textsc{Max-LPA}, does
not suffer from this problem because
Poljak and S\r{u}ra~\cite{poljak1983periodical} have shown in a 
different context that \textsc{Max-LPA} converges to a cycle of 
period 1 or~2.

Despite the simplicity of LPA, there has been
very little formal analysis of either the convergence time of LPA or the quality of
communities produced by it. There have been papers~\cite{raghavan2007near,
leung2008towards, cordasco2010community} that provide some empirical results
about LPAs. For example, the number of iterations of label updates required
for the correct convergence of LPA is around 5~\cite{raghavan2007near},  but it
is hard to
derive any fundamental conclusions about LPA's behavior, 
even on specific families of networks, from these empirical results. 
One reason for this state of affairs is that despite its simplicity, even on simple networks, LPA can have
complicated behavior, not unlike epidemic processes that model the spread of
disease in a networked population~\cite{newman2002disease}.
Our goal in this paper is to initiate a systematic analysis of the behavior
of \textsc{Max-LPA}, both in terms of its convergence time and in terms
of the ``quality'' of communities produced.

Watts and Strogatz \cite{watts1998model} have pointed out that 
the classical \er model of random graphs differs from real-world 
social, technological, and biological networks in several critical ways.
Following this, a variety of other random graph models have
been considered as models of real-world networks.
These include the \textit{configuration} model
\cite{molloy1995critical,bender1978config},
the \textit{Watts-Strogatz} model \cite{watts1998model}, 
\textit{preferential attachment} models \cite{barabasi2002model}, etc.
(for definitions and more examples, see~\cite{newman2010networks}).
There is no empirical study or formal analysis of LPAs on these
classes of networks. 
As our first step towards developing analysis techniques for LPAs we define 
a clustered version of \er random graphs and present a formal proof of 
the running times of LPAs on these networks. 
We realize that \er networks and even clustered \er networks are 
inadequate models of real world networks, but believe that our analysis 
techniques could be useful in general.

The variants of LPA can naturally be viewed as \textit{distributed} 
algorithms, meaning each node only has \textit{local}
knowledge, i.e., knowledge of its label and the labels of its 
neighbors obtained by means of
message passing along edges of the networks. Distributed algorithms
are generally classified as \textit{synchronous} or \textit{asynchronous}
algorithms. (The reader is referred to standard books
(e.g.,~\cite{peleg2000distributed}) for a full exposition of these
terms). 
Here we analyze a synchronous version of \textsc{Max-LPA}.
The algorithm proceeds in
rounds  and in each round each node sends its label to all neighbors and then
updates its label based on the labels received from neighbors and its own 
label.

\subsection{Preliminaries} \label{sub:pre}
We use $G = (V,E)$ to denote an undirected connected graph (network) of size
$n=|V|$.
For $v \in V$, we denote by $N(v) = \{u : u \in V, (u,v) \in E\}$ the neighborhood
of $v$ in graph $G$, by $deg(v) = |N(v)|$ the degree of $v$, and by
$\Delta(G)=\max_{v\in V} deg(v)$ the maximum degree over all the
vertices in $G$.
A \textit{$k$-hop neighborhood} ($k \geqslant 1$) of $v$ is defined as $N_k(v) = \{w : \mbox{dist}_G(w, v) \le k\} \setminus \{v\}$. 
We denote the \textit{closed neighborhood} (respectively,
\textit{closed $k$-hop neighborhood}) of $v$ as $N'(v) = N(v) \cup \{v\}$
(respectively, $N'_k(v) = N_k(v) \cup \{v\}$).

Denote by $\ell_u(t)$ the label of node $u$ just before round $t$.
When the round number is clear from the context, we
use $\ell_u$ to denote the current label of $u$. 
Since the number of labels in the network is finite,
LPA will behave periodically starting in some round $t^*$, i.e., 
for some $p \ge 1$, $0 \le i < p$, and $j = 0, 1, 2, \ldots$,
$$\ell_u(t^*+ i)=\ell_u(t^* + i + j \cdot p)$$ 
for all $u \in V$. 
Then we say that \textsc{Max-LPA} has \textit{converged} in $t^*$ rounds. 

We now describe \textsc{Max-LPA} precisely (see \textbf{Algorithm 1}). 
Every node $v \in V$ is assigned a unique label uniformly and independently
at random. 
For concreteness, we assume that these labels come from the range $[0, 1]$.
At the start of a round, each node sends its label to all neighboring nodes. 
After receiving labels from all neighbors, a node $v$ updates its label as:
\begin{equation} \label{eqn:label}
l_v  \leftarrow \max\left\{\ell \mid \sum_{u \in N'(v)}[\ell_u== \ell] \ge \sum_{u \in N'(v)}[\ell_u== \ell']\mbox{ for all $\ell'$}\right\},
\end{equation}
where $[\ell_u==\ell]$ evaluates to 1 if $\ell_u=\ell$,
otherwise evaluates to 0.
Note that there is no randomness in the algorithm after the initial
assignments of labels. 
\begin{algorithm}[t]
\caption{\textsc{Max-LPA} on a node $v$} \label{algo:lpa}
\begin{algorithmic}
 \STATE $i=0$
 \STATE $l_v[i] \leftarrow$ random(0,1)
 \WHILE{true}\label{algo:lpa:while}
   \STATE $i++$;
   \STATE send $l_v[i-1]$ to $\forall u \in N(v)$
   \STATE receive $l_u[i-1]$ from $\forall u \in N(v)$
   \STATE $l_v[i]  \leftarrow \max\left\{\ell \mid \sum_{u \in
N'(v)}[\ell_u[i-1]== \ell] \ge \sum_{u \in N'(v)}[\ell_u[i-1]== \ell']\mbox{ for
all $\ell'$}\right\}$
 \ENDWHILE
\end{algorithmic}
\end{algorithm}


By ``w.h.p.'' (with high
probability)
we mean with probability at least $1-\frac{1}{n^c}$ for some constant
$c\geqslant1$. 
In this paper we repeatedly use the following versions of a tail
bound on the probability
distribution of a random variable, due to Chernoff and Hoeffding
\cite{chernoff1952measure, hoeffding1963probability}.
Let $X_1, X_2, \ldots, X_m$ be independent and identically distributed binary
random variables.
Let $X = \sum_{i = 1}^m X_i$.
Then, for any $0 \le \epsilon \le 1$ and $c \geqslant 1$,
\begin{align}
\label{eqn:chernoff1}
  \Pr\left[X > (1 + \epsilon) \cdot E[X] \right] & \le 
\exp\left(-\frac{\epsilon^2 E[X]}{3}\right)\\
\label{eqn:chernoff3}
  \Pr\left[X < (1 - \epsilon)\cdot E[X]\right]  & \le 
\exp \left(-\frac{\epsilon^2 E[X]}{2}\right) \\
\label{eqn:chernoff2} 
  \Pr\left[|X - E[X]| > \sqrt{ 3c \cdot E[X] \cdot \log n}\right] & \le 
\frac{1}{n^c} 
\end{align}

\subsection{Results} \label{sub:result}
As mentioned earlier, the purpose of this paper is to counterbalance the predominantly
empirical line of research on LPA and initiate a systematic analysis of \textsc{Max-LPA}.
Our main results can be summarized as follows:
\begin{itemize}
\item As a ``warm-up'' we prove (Section~\ref{sec:path}) 
that when executed on an $n$-node path \textsc{Max-LPA} 
converges to a cycle of period one in $\Theta(\log n)$ rounds w.h.p.
Moreover, we show that w.h.p. the state that \textsc{Max-LPA} converges to 
has $\Omega(n)$ communities.

\item %
In our main result (Section~\ref{sec:er}), we define a class of random graphs
that we call \textit{clustered \er graphs}. 
A clustered \er graph $G = (V, E)$ comes with a node partition $\Pi = (V_1, V_2, \ldots, V_k)$
and pairs of nodes in each $V_i$ are connected with probability $p_i$ and pairs 
of nodes in distinct parts in $\Pi$ are connected with probability $p' < \min_i \{p_i\}$.
Since $p'$ is small relative to any of the $p_i$'s, one might view a clustered \er graph
as having a natural community structure given by $\Pi$.
We prove that even with fairly general restrictions on the $p_i$'s and $p'$ and
on the sizes
of the $V_i$'s, \textsc{Max-LPA} converges to a period-1 cycle in just 2 rounds, w.h.p.
\textit{and} ``correctly'' identifies $\Pi$ as the community structure of $G$.

\item Roughly speaking, the above result requires each $p_i$ to be 
$\Omega\left(\left(\frac{\log n}{n}\right)^{1/4}\right)$.
We believe that \textsc{Max-LPA} would correctly and quickly identify $\Pi$
as the community structure of a given clustered \er graph even when the
$p_i$'s are much smaller, e.g. even when $p_i = \frac{c \log n}{n}$ for $c > 1$.
However, at this point our analysis techniques do not seem adequate for situations with
smaller $p_i$ values and so we provide empirical evidence (Section~\ref{sec:erp}) 
for our conjecture that \textsc{Max-LPA} correctly converges to $\Pi$ in
$O(\mbox{polylog}(n))$ rounds even when $p_i = \frac{c \log n}{n}$ for some 
$c > 1$ and $p'$ is just a logarithmic factor smaller than $p_i$.
\end{itemize}

\subsection{Related Work}\label{sub:relwork}
There are several variants of LPA presented in the
literature~\cite{cordasco2010community, gregory2009finding,
subelj2011unfolding, liu2009bipartite}. Most of these are concerned about
``quality'' of the output and present 
empirical studies of output produced by LPA. 

Raghavan et al.~\cite{raghavan2007near}, based on the experiments, claimed that
at least 95\% of the nodes are classified correctly by the end of 5 rounds of label updates. But the
experiments that they carried out were on the small networks. 

Cordasco and Gargano~\cite{cordasco2010community} proposed a semi-synchronous
approach which is guaranteed to converge without oscillations and can be parallelized. 
They provided a formal proof of convergence but did bound the 
running time of the algorithm. Lui and
Murata~\cite{liu2009bipartite} presented a variation of LPA for
bipartite networks which converges but no formal proof is provided, 
neither for the convergence nor for the running time. 

Leung et al.~\cite{leung2008towards} presented
empirical analysis of quality of output produced by LPA on larger data sets.
From experimental results on a special structured network they claimed that 
running time of LPA is  $O(\log n)$.

\section{Analysis of \textsc{Max-LPA} on a Path} \label{sec:path}
Consider a path $\mathcal{P}_n$ consisting of vertices
$V = [n]$ and edge set $E=\{(i, i+1) \mid 1 \le i < n\}$.
In this section, we analyze the execution of \textsc{Max-LPA} on a path network 
$\mathcal{P}_n$ and prove that in $O(\log n)$ rounds \textsc{Max-LPA}
converges to a state from which no further label updates occur and
furthermore in such a state the number of communities is
$\Omega(n)$ w.h.p..

\begin{lemma}
\label{lemma:periodOne}
When \textsc{Max-LPA} is executed on path network $\mathcal{P}_n$, independent
of 
the initial label assignment, it will converge
to a state from which no further label updates occur.
\end{lemma}
 \begin{proof}
 First we show that at any point in the execution of \textsc{Max-LPA}, the
 subgraph of $\mathcal{P}_n$
 induced by all nodes with the same label, is a single connected component.
 This is true before the first round since the initial label assignment assigns
 distinct labels to the nodes.
 Suppose the claim is true just before round $t$.
 Let $S = (i, i+1, \ldots, j)$ be the subgraph of $\mathcal{P}_n$ consisting of
 nodes with label $\ell$, just before round $t$ of \textsc{Max-LPA}.
 \begin{itemize}
 \item If $S$ contains two or more nodes then none of the nodes in $S$ will ever
 change
 their label.
 Moreover, the only other nodes that can acquire label $\ell$ in round $t$ are
 nodes 
 $i-1$ and $j+1$.
 Hence, after round $t$, the set of nodes with label $\ell$ still induces a
 single 
 connected component.
 
 \item If $S$ contains a single node, say $i$, then the only way in which 
 label $\ell$ might induce multiple connected components after round $t$
 would be if in round $t$: (i) node $i-1$ acquires label $\ell$,
 (ii) node $i+1$ acquires label $\ell$, and (iii) node $i$ changes its label to
 some $\ell' \not= \ell$.
 (i) and (ii) above can only happen if $\ell$ is larger than the labels of nodes
 $i-1$ and $i+1$ just
 before round $t$. But, if this were true, then node $i$ would not change its
 label 
 in round $t$.
 \end{itemize}
 Hence, in either case the nodes with label $\ell$ would induce a connected
 component.
 
 According to Poljak and S\r{u}ra~\cite{poljak1983periodical}, \textsc{Max-LPA}
 has a period of 1 or 2 on
 any network with any initial label assignment.
 To obtain a contradiction we suppose that \textsc{Max-LPA} has a period of 2
 when 
 executed on $\mathcal{P}_n$ for some $n$ and some initial label 
 assignment. 
 Therefore for some $v \in V$ and some time $t$, $\ell_v(t + 2i) = \ell$
 and $\ell_v(t + 2i + 1) = \ell'$ for $\ell \not= \ell'$ and all $i = 0,
 1, 2, \ldots$.
 For $v$ to change its label from $\ell$ to $\ell'$ in a round it must
 be the case that $\ell < \ell'$.
 This is because $v$ cannot have two neighbors with label $\ell'$ since $\ell'$
 can only induce one connected component.
 Hence, $v$ acquires the new label $\ell'$ by tie breaking.
 By a symmetric argument, for $v$ to change its label from $\ell'$ to
 $\ell$ in the next round, it must be the case that $\ell' < \ell$.
 Both conditions cannot be met and we have a contradiction.
 \qed
 \end{proof}

\begin{definition}\label{def:local}
A node $v$ is said to be \textit{$k$-hop maxima} if its label $\ell_v$ 
is (strictly) greater than the labels of all nodes in its $k$-neighborhood.
As a short form, we will use \textit{local maxima} to refer to any node
that is a 1-hop maxima.
\end{definition}

Let $M = \{i_1, i_2,\ldots, i_r\}$, $i_1 < i_2 \cdots < i_r$ be the set of 
nodes which are 2-hop maxima in $\mathcal{P}_n$ for the given initial 
label assignment. 
For any $1 \le j < r$, nodes $i_j$ and $i_{j+1}$ are said to be 
\textit{consecutive} 2-hop maxima.

\begin{lemma}
\label{lemma:communities}
When \textsc{Max-LPA} converges, the number of communities it identifies is
bounded
below by the number of 2-hop maxima in the initial label assignment.
\end{lemma}
 \begin{proof}
 Since all initial node labels are assumed to be
 distinct, in the first round of \textsc{Max-LPA}
 every node $u \in V$ acquires a label by breaking ties. 
 Since ties are broken in favor of larger labels, all neighbors of
 each $i_j\in M$ will acquire the corresponding 2-hop maxima label
 $\ell_{i_j}$. 
 Thus after one round of \textsc{Max-LPA}, for each $i_j \in M$, there are three
 consecutive nodes in $\mathcal{P}_n$ with label $\ell_{i_j}$.
 None of these nodes will change their label in future rounds and
 hence there will be a community induced by label $\ell_{i_j}$
 when \textsc{Max-LPA} converges.
 \qed
 \end{proof}

\begin{lemma}
\label{lemma:convergence}
Let $D$ be the maximum distance in $\mathcal{P}_n$ between a pair of consecutive
nodes in $M$.
Then the number of rounds that \textsc{Max-LPA} takes to converge is bounded
above by $D + 2$.
\end{lemma}
 \begin{proof}
 Call a node $v$ \textit{isolated} if its label is distinct from the labels of
 its neighbors.
 After the first round of \textsc{Max-LPA} each node $i_j \in M$ and its
 neighbors 
 acquire label $\ell_{i_j}$.
 Therefore, after the first round, every connected component of the graph
induced
 by isolated nodes has size bounded above by $D$.
 We now show that in each subsequent round, the size of every connected
component
 of size two or more will reduce by at least one.
 Let $S$ be a component in the graph induced by isolated nodes, just before
round
 $t$.
 Let $i$ be the node with maximum label in $S$. Since $S$ contains at least two
 nodes,
 without loss of generality suppose that $i+1$ is also in $S$.
 In round $t$, node $i$ could acquire the label of a node outside $S$.
 If this happens $i$ would cease to be isolated after round $t$.
 Similarly, in round $t$, node $i+1$ could acquire the label of a node outside
 $S$ and would therefore
 cease to be isolated after round $t$.
 If neither of these happens in round $t$, then node $i+1$ will acquire the
label
 of node $i$ in round $t$
 and node $i$ will not change its label.
 In this case, both $i$ and $i+1$ will cease to be isolated nodes after round
 $t$.
 In any case, we see that the size of the component $S$ has shrunk by at least
 one
 in round $t$.
 Thus in $D+1$ rounds $\mathcal{P}_n$ we will reach a state in which all
 components in
 the graph induced by isolated nodes have size one.
 Isolated nodes whose labels are larger than the labels of neighbors will make
no
 further
 label updates.
 The remaining isolated nodes will disappear in one more round.
 \qed
 \end{proof}

\begin{theorem} \label{theorem:path}
When \textsc{Max-LPA} is executed on a path $\mathcal{P}_n$, 
it converges to a state from which no further label updates occur in 
$O(\log n)$ rounds w.h.p.
Furthermore, in such a state there are $\Omega(n)$ communities.
\end{theorem}
\begin{proof}
Partition $\mathcal{P}_n$ into ``segments'' of 5 nodes each. 
Let $S$ denote the set of center nodes of these segments.
The probability that a node in $\mathcal{P}_n$ is a 2-hop 
maxima is $\frac{1}{5}$.
Therefore the expected number of nodes in $S$ that end up being 2-hop
maxima is $n/25$.
Now note that for any two nodes $i, j \in S$, node $i$ being a 2-hop maxima
is independent of node $j$ being a 2-hop maxima due to the fact that there are
at least 4 nodes between $i$ and $j$.
Therefore, we can apply the lower tail Chernoff bound (\ref{eqn:chernoff3}) to
conclude 
that w.h.p. at least $n/50$ nodes in $\mathcal{P}_n$ are
2-hop maxima.
Combining this with Lemma~\ref{lemma:communities} tell us that 
when \textsc{Max-LPA} converges, it does so to a state in which 
there are at least $n/50$ communities with high probability.

Now consider a contiguous sequence of $k$ 5-node segments.
The probability that none of the centers of the $k$ segments are 2-hop maxima
is $(4/5)^k$.
Note that here we use the independence of different segment centers becoming
2-hop maxima.
Hence, for a large enough constant $c$, the probability that none of the 
centers of $k = c \log n$ consecutive segments are 2-hop maxima is at most
$1/n^2$.
Using the union bound and observing that there at most $n$ consecutive
segment sequences of length $k$, we see that the probability that there is
a sequence of $k = c \log n$ consecutive segments, none of whose centers
are 2-hop maxima, is at most $1/n$.
Therefore, with probability at least $1 - 1/n$ every sequence of $k = 
c \log n$ consecutive segments contains a segment whose center is a 
2-hop maxima. 
It follows that the distance between consecutive 2-hop maxima is at most
$5c \log n$ with probability at least $1 - 1/n$.
The result follows by combining this with Lemma~\ref{lemma:convergence}. 
\qed
\end{proof}

The argument given here establishing a linear lower bound on the number of communities can 
be easily extended to graphs with maximum degree bounded by a constant.
The argument bounding the convergence time depended crucially on two properties of the underlying graph:
(i) degrees being bounded and (ii) number of paths of length $O(\log n)$ being 
polynomial in number. 
Thus the convergence bound can be extended to other graph classes 
satisfying these two properties (e.g., trees with bounded degree).

\section{Analysis of \textsc{Max-LPA} on Clustered \er Graphs}
\label{sec:er}
We start this section by introducing a family of ``clustered'' random
graphs that come equipped with a simple and natural notion of 
a community structure.
We then show that on these graphs \textsc{Max-LPA} detects this natural
community structure in
only 2 rounds, w.h.p. provided certain fairly general
sparsity conditions are satisfied.

\subsection{Clustered \er graphs} \label{sub:structure}
Recall that for an integer $n \ge 1$ and $0 \le p \le 1$, the \er graph
$G(n, p)$ is the random graph obtained by starting with vertex set 
$V = \{1,2,\ldots,n\}$ and connecting each pair of vertices $u, v \in V$,
independently with probability $p$.  
Let $\Pi$ denote a partition $(V_1, V_2, \ldots, V_k)$ of $V$,
let $\pi$ denote the real number sequence $(p_1, p_2, \ldots, p_k)$,
where $0 \le p_i \le 1$ for all $i$ and let $0 \le p' < \min_i \{p_i\}$.
The \textit{clustered \er} graph $G(\Pi, \pi, p')$ has vertex set
$V$ and edges obtained by independently connecting each pair of
vertices $u, v \in V$ with probability $p_i$ if $u, v \in V_i$ for some
$i$ and with probability $p'$, otherwise (see Figure~\ref{fig:erdos}).
Thus each induced subgraph $G[V_i]$ is the standard \er graph 
$G(n_i, p_i)$, where $n_i = |V_i|$.
\begin{figure}[t]
\centering
\fbox{ 
 \begin{tikzpicture}[scale=1]
  \draw (2,0) ellipse (1 and 2);
  \node[blank] at (2,-2.5) {$V_1$};
  \node[cir] at (2.2,1) (u1) {$u_1$};
  \node[cir] at (1.8,-0.3) (u2) {$u_2$};
  \node[cir] at (2.2,-1.2) (u3) {$u_{n_1}$};
  \node[blank] at (2.4,0) {$\vdots$};
  \draw (u1) -- (u2) node[above,left,midway]{$p_1$};
  \draw (u2) -- (u3) node[above,left,midway]{$p_1$};
  
  \draw (6,0) ellipse (1 and 2);
  \node[blank] at (6,-2.5) {$V_2$};
  \node[cir] at (5.6,1) (v1) {$v_1$};
  \node[cir] at (6,0) (v2) {$v_2$};
  \node[cir] at (5.8,-1) (v3) {$v_{n_2}$};
  \node[blank] at (6.5,0.5) {$\vdots$};
  \draw (v1) -- (v2) node[above,left,midway]{$p_2$};
  \draw (v2) -- (v3) node[above,left,midway]{$p_2$}; 

  \node[blank] at (8.5,0) {$\ldots$};
  \node[blank] at (8.5,-2.5) {$\ldots$};

  \draw (11,0) ellipse (1 and 2);
  \node[blank] at (11,-2.5) {$V_k$};
  \node[cir] at (11.2,1) (w1) {$w_1$};
  \node[cir] at (10.7,-0.4) (w2) {$w_2$};
  \node[cir] at (11.2,-1.2) (w3) {$w_{n_k}$};
  \node[blank] at (10.4,0.3) {$\vdots$};
  \draw (w1) -- (w3) node[above,right,midway]{$p_k$};
  \draw (w2) -- (w3) node[above,left,midway]{$p_k$}; 

   \draw (u1) -- (v2) node[below=5pt,left,midway]{$p'$};
 \draw (v3) -- (w2) node[below=5pt,right,midway]{$p'$};
\end{tikzpicture}
 }
\caption{The clustered \er graph. We connect two nodes in the $i$-th ellipse 
(i.e., $V_i$) with probability $p_i$ and nodes from different ellipses are
connected
with probability $p'< \min_i\{p_i\}$. \label{fig:erdos}}
\end{figure}
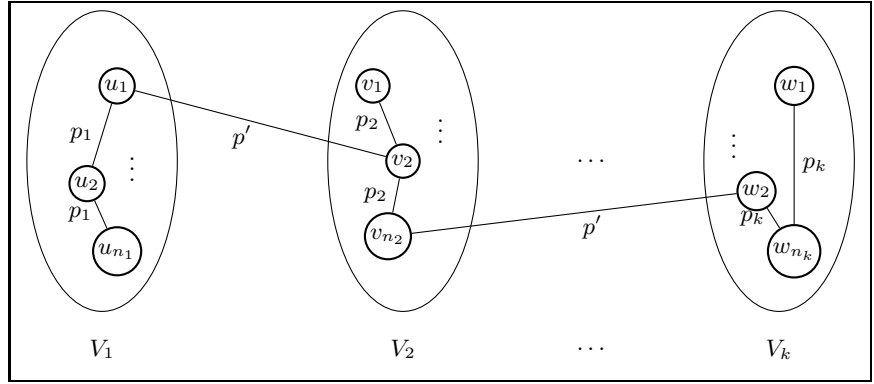

Given that $p' < p_i$ for all $i$, one might view $G(\Pi, \pi, p')$ as having a
natural
community structure given by the vertex partition $\Pi$.
Specifically, when $p'$ is much smaller than $\min_i\{p_i\}$, the
inter-community edge density is 
much less than the intra-community edge density and it may be easier
to detect the community structure $\Pi$.
On the other hand as the intra-community probabilities $p_i$ get closer to 
$p'$, it may be hard for an algorithm
such as  \textsc{Max-LPA} to identify $\Pi$ as the community structure.
Similarly, if an intra-community probability $p_i$ becomes very 
small, then the subgraph $G[V_i]$
can itself be quite sparse and independent of how small $p'$ is relative to 
$p_i$, any community detection algorithm may end up viewing each $V_i$ as
being composed of several communities.

In the rest of the section, we explore values of the $p_i$'s and $p'$ for which
\textsc{Max-LPA} ``correctly'' and quickly identifies $\Pi$ as the community
structure of
$G(\Pi, \pi, p')$.

\subsection{Analysis} 
\label{subsection:analysis}
In the following theorem we establish fairly general conditions on the
probabilities
$\{p_i\}$ and $p'$ and on the node subset sizes $\{n_i\}$ and $n$ under which
\textsc{Max-LPA} converges
correctly, i.e., to the node partition $\Pi$, w.h.p.
Furthermore, we show that under these circumstances just 2 rounds suffice for
\textsc{Max-LPA} to reach convergence!

\begin{lemma}
\label{lemma:probBound}
Let $G(\Pi, \pi, p')$ be a clustered \er graph such that $p' <
\min_i\{\frac{n_i}{n}\}$.
Let $\ell_i$ be the maximum label of a node in $V_i$.
Then for any node $v \in V_i$ the probability that $v$ is not adjacent to
a node outside $V_i$ with label higher than $\ell_i$ is at least $1/2e$.
\end{lemma}
\begin{proof}
Let $v'$ be a node in $V \setminus V_i$.
Given that $|V_i| = n_i$ and $\ell_i$ is the maximum label among these $n_i$
nodes, the probability that the label assigned uniformly at random to $v'$ is
larger than $\ell_i$ is $1/(n_i + 1)$.
The probability that $v$ has an edge to $v'$ \textit{and} $v'$ has a higher
label than $\ell_i$ is $p'/(n_i + 1)$.
Therefore the probability that $v'$ has no edge to a node outside $V_i$ with
label larger than $\ell_i$ is 
$$\left(1 - \frac{p'}{n_i + 1}\right)^{n - n_i}.$$
We bound this expression below as follows: 
$$\left(1 - \frac{p'}{n_i + 1}\right)^{n - n_i}  > \left(1 -
\frac{p'}{n_i}\right)^{n} > \left(1 - \frac{1}{n}\right)^{n} > \frac{1}{2e}.$$
\qed
\end{proof}

\begin{theorem}
\label{theorem:ER}
Let $G(\Pi, \pi, p')$ be a clustered \er graph. 
Suppose that the probabilities $\{p_i\}$ and $p'$ and the node subset sizes
$\{n_i\}$ and $n$
satisfy the inequalities: 
$$\mbox{(i)   } n_i p_i^2 > 8n p' \qquad\mbox{and}\qquad \mbox{(ii)   } n_i
p_i^4 > 1800 c \log n,$$
for some constant $c$.
Then, given input $G(\Pi, \pi, p')$, \textsc{Max-LPA} converges correctly to
node partition $\Pi$ in two rounds w.h.p. (Note that condition (ii) implies for
each $i$, $p_i > \frac{\log n_i}{n_i}$ and hence each $G[V_i]$ is connected.)
\end{theorem}
\begin{proof}
Let $V_i = \{u_1, u_2, \ldots, u_{n_i}\}$ and without loss of generality assume
that $\ell_{u_1} 
> \ell_{u_2} > \cdots > \ell_{u_{n_i}}$.
Since all initial node labels are assumed to be distinct, in the first round of
\textsc{Max-LPA}
every node $u \in V$ acquires a label by breaking ties. 
Since ties are broken in favor of larger labels, all neighbors of $u_1$ in $V_i$
that have no neighbor outside $V_i$ with a label larger than $\ell_{u_1}$ 
will acquire the label $\ell_{u_1}$.
Consider a node $v \in V_i$.
Let $\beta$ denote the probability that $v$ has no neighbor outside $V_i$
with label larger than $\ell_{u_1}$.
Note that inequality (i) in the theorem statement implies the hypothesis of
Lemma \ref{lemma:probBound} and therefore $\beta > 1/2e$.
The probability that $v$ is a neighbor of $u_1$ and does not have a neighbor
outside $V_i$ is $\beta \cdot p_i$.
Hence, after the first round of \textsc{Max-LPA}, in expectation, $n_i \cdot
\beta \cdot p_i$ nodes
in $V_i$ would have acquired the label $\ell_{u_1}$.
In the rest of the proof we will use
$$X := n_i \cdot \beta \cdot p_i.$$

Now consider node $u_j$ for $j > 1$.
For a node $v \in V_i$ to acquire the label $\ell_{u_j}$ it must be the case
that $v$ is adjacent to $u_j$, not adjacent to any node in 
$\{u_1, u_2, \ldots, u_{j-1}\}$, and not adjacent to any node outside $V_i$
with a label higher than $\ell_{u_j}$.
Since $\ell_{u_j}$ is smaller than $\ell_{u_1}$, the probability that $v$
is not adjacent to a node outside $V_i$ with label higher than $\ell_{u_j}$
is less than $\beta$.
Thus the probability that a node in $V_i$ acquires the label $\ell_{u_j}$ is at
most $p_i (1 - p_i)^{j-1} \cdot \beta < p_i (1 - p_i) \cdot \beta$.
Furthermore, the probability that a node outside $V_i$ will acquire 
the label $\ell_{u_j}$ at the end of the first round is at most $p'$.
Therefore, the expected number of nodes in $V$ that acquire the label $u_j$, at
the end of the first round,
is in expectation at most $n_i \cdot p_i (1 - p_i) \cdot \beta + (n - n_i)p'$.
We now use inequality (i) and the fact that $2\beta e > 1$ to upper bound this
expression as follows:
$$n_i \cdot p_i(1 - p_i) \cdot \beta + (n - n_i)p' < n_i \cdot p_i(1 - p_i)
\cdot \beta + \frac{2 \beta e \cdot n_i p_i^2}{8} <
n_i \cdot p_i\left(1 - \frac{3p_i}{4}\right) \cdot \beta.$$
Therefore, the expected number of nodes in $V$ that acquire the label $u_j$, at
the end of the first round,
is in expectation at most 
$$Y := n_i \cdot p_i\left(1 - \frac{3p_i}{4}\right) \cdot \beta.$$ 
It is worth mentioning at this point that $X - Y = n_i p_i^2 \beta/4$.

Note that all the random variables we have utilized thus far, e.g., the number
of nodes adjacent to $u_1$
and not adjacent to any node outside $V_i$ with label higher than $\ell_{u_1}$,
can be expressed as sums 
of independent, identically distributed indicator random variables.
Hence we can bound the deviation of such random variables using the tail bound
in (\ref{eqn:chernoff2}).
In particular, let $Y'$ denote $Y + \sqrt{3c Y \log n}$ and $X'$ denote $X -
\sqrt{3c X \log n}$.
With high probability, at the end of the first round of \textsc{Max-LPA}, the
number of nodes in $V_i$ 
that acquire the label $u_1$ is at least $X'$ and the number of nodes in $V$
that acquire the label $\ell_{u_j}$, $j > 1$,
is at most $Y'$.
Next we bound the ``gap'' between $X'$ and $Y'$ as follows:
\begin{align*}
X' - Y' & =  X - Y - \sqrt{3c X \log n} - \sqrt{3c Y \log n}\\
        & >  \frac{3n_i p_i^2 \beta}{4} - 2 \sqrt{3c X \log n}\\
        & >  \frac{3n_i p_i^2 \beta}{4} - 2 \sqrt{3c n_i p_i \beta \log n}\\
        & >  \frac{3n_i p_i^2 \beta}{4} - \frac{3n_i p_i^2 \beta}{5}\\
        & =  \frac{3 n_i p_i^2 \beta}{20} 
\end{align*} 
The second inequality follows from $X - Y = 3n_ip_i^2\beta/4$ and $Y < X$,
the third from the fact that $X = n_i p_i \beta$, and the fourth by using
inequality (ii) from
the theorem statement.

We now condition the execution of the second round of \textsc{Max-LPA} on the
occurrence of the two high probability
events: (i) number of nodes in $V_i$ that acquire the label $u_1$ is at least
$X'$ and 
(ii) the number of nodes in $V$ that acquire the label $\ell_{u_j}$, $j > 1$,
is at most $Y'$.
Consider a node $v \in V_i$ just before the execution of the second round of
\textsc{Max-LPA}.
Node $v$ has in expectation at least $p_i X'$ neighbors labeled $\ell_{u_1}$ in
$V_i$.
Also, node $v$ has in expectation at most $p_i Y'$ neighbors labeled
$\ell_{u_j}$, for each $j > 1$, in $V$.
Let us now use $X''$ to denote the quantity $p_i X' - \sqrt{3c p_i X' \log n}$
and
$Y''$ to denote the quantity $p_i Y' + \sqrt{3c p_i Y' \log n}$.
By using (\ref{eqn:chernoff2}) again, we know that w.h.p. $v$ has
at least
$X''$ neighbors with label $\ell_{u_1}$ and at most $Y''$ neighbors with a label
$\ell_{u_j}$, $j > 1$.
We will now show that $X'' > Y''$ and this will guarantee that in the second
round of \textsc{Max-LPA} $v$ will acquire the label $\ell_{u_1}$, with high
probability. Since $v$ is an
arbitrary node in $V_i$, this implies that all nodes in $V_i$ will acquire the
label $\ell_{u_1}$ 
in the second round of \textsc{Max-LPA} w.h.p.
\begin{align*}
X'' - Y'' & =  p_i(X' - Y') - \sqrt{3c p_i X' \log n} - \sqrt{3c p_i Y' \log
n}\\
        & >  \frac{3 n_i p_i^3}{20} - 2 \sqrt{3c p_i X' \log n}\\
        & >  \frac{3 n_i p_i^3}{20} - 2 \sqrt{3c n_i p_i^2 \beta \log n}\\
        & >  \frac{3 n_i p_i^3}{20} - \frac{n_i p_i^3 \beta}{10}\\
        & =  \frac{3 n_i p_i^2}{20}\\
        & >  0
\end{align*}
The second inequality follows from the bound on $X' - Y'$ derived earlier and
$Y' < X'$,
the third from the fact that $X' < n_i p_i \beta$, and the fourth by using
inequality (ii) from
the theorem statement.

Thus at the end of the second round of \textsc{Max-LPA}, w.h.p.,
every node in $V_i$ has label
$\ell_{u_1}$. This is of course true, w.h.p., for all of the
$V_i$'s.
Now note that every node $v \in V_i$ has, in expectation $n_i p_i$ neighbors in
$V_i$ and
fewer than $n p'$ neighbors outside $V_i$.
Inequality (i) implies that $n p' < n_i p_i/8$ and inequality (ii) implies that
$n_i p_i = \Omega(\log n)$.
Pick a constant $\epsilon > 0$ such that $n_i p_i (1 + \epsilon)/8 < n_i p_i (1
- \epsilon)$.
By applying tail bound (\ref{eqn:chernoff1}), we see that w.h.p.
$v$ has more than
$n_i p_i (1 - \epsilon)$ neighbors in $V_i$ and fewer than $n_i p_i (1 +
\epsilon)/8$ neighbors
outside $V_i$. Hence, w.h.p. $v$ has no reason to change its label.
Since $v$ is an arbitrary node in an arbitrary $V_i$, w.h.p. there are no
further changes
to the labels assigned by \textsc{Max-LPA}.
\qed
\end{proof}

To understand the implications of Theorem \ref{theorem:ER} consider the
following example.
Suppose that the clustered \er graph has $O(1)$ clusters and each cluster had
size $\Theta(n)$.
In such a setting, inequality (ii) from the theorem simplifies to requiring that
each $p_i = \Omega((\log n/n)^{1/4})$
and inequality (ii) simplifies to $p' < p_i^2/c$ for all $i$.
This tells us, for instance, that \textsc{Max-LPA} converges in just two rounds
on a clustered \er graph in which
each cluster has $\Theta(n)$ vertices and an intra-community probability of
$\Theta(1/n^{1/3})$ and
the inter-community probability is $\Theta(1/n^{2/3})$.

This example raises several questions. 
If we were willing to allow more time for \textsc{Max-LPA} to converge, say
$O(\log n)$ rounds, could we significantly
weaken the requirements on the $p_i$'s and $p'$.
Specifically, could we permit an intra-community probability $p_i$ to become as small as 
$c \log n/n$ for some constant $c > 1$?
Similarly, could we permit the inter-community probability $p'$ to come much closer to the smallest
$p_i$, say within a constant factor.

We believe that it may be possible to obtain such results, but only via substantively different
analysis techniques.

 \section{Empirical Results on Sparse \er Graphs}
\label{sec:erp}
In the previous section we proved that if the clusters (each $V_i$) in a
clustered \er graphs were dense enough and the inter-cluster edge density
(fraction of edges between nodes in different $V_i$) was relatively low, then
\textsc{Max-LPA} would
correctly converge in just 2 rounds. Specifically, our result requires each
cluster to be \er random graph $G(n,p)$ with $p=O\left(\left(\frac{\log n}{n}\right)^{1/4}\right)$.
In this section we ask: how
does \textsc{Max-LPA} behave if individual clusters are much sparser? For
example, how
does \textsc{Max-LPA} behave on $G(n,p)$ with much smaller $p$, say
$p=\frac{c\cdot
\log n}{n}$ for some $c>1$. The proof technique used in the previous section
does not extend to such small values of $p$. However, we believe that
\textsc{Max-LPA}
converges quickly and correctly even on clustered \er graphs whose clusters are
of the type $G(n,p)$ for $p = \frac{c\cdot \log n}{n}$ for $c>1$. In this
section, 
we ask (and empirically answer) two questions:
\begin{enumerate}
\item Can one expect there to be a constant $c$ such that \textsc{Max-LPA}, when
run
on $G(n,p)$ with $p \ge \frac{c \log n}{n}$ will, with high probability,
terminate with one community.
If the answer to Question 1 is ``yes'' what might the running time of 
\textsc{Max-LPA}, as a function of $n$ be for appropriate values of $p$.

\item Consider a clustered \er graph with two parts $V_1$ and $V_2$ of equal size,
and each $p_i  = \frac{c \log n}{n}$ for some $c > 1$.
Let $p' = \frac{c'}{n}$ for some $c'$.
Are there constants $c, c'$ for which \textsc{Max-LPA} will quickly converge
and correctly identify $(V_1, V_2)$ as the community structure?
\end{enumerate}
We are interested in values of $p$ of the form $\frac{c\cdot \log n}{n}$ because
$\frac{\log n}{n}$ is the threshold for
connectivity in \er graphs~\cite{erdos1960evolution}.

 \subsection{Simulation Setup}
  We implemented \textsc{Max-LPA} in a C program and executed on a Linux machine
 (with 2.4
 GHz Intel(R) Core(TM)2 processor). We examined the number of rounds it takes
and
 also number of communities it declares at the end of the execution.
 We executed \textsc{Max-LPA} on $G(n,p)$ and on $G(\Pi, \pi, p')$ with
 $\Pi = (V_1, V_2)$, $|V_1| = |V_2| = n/2$, $\pi = (p, p)$, $p'=0.6/n$ for
various values of $n$ and $p$.
 For each $n$, $p$ combination we ran \textsc{Max-LPA} 50 times. We used $p$
values of
 the form $\frac{c\cdot \log n}{n}$ for various values of $c\geq 1$.
 
 \subsection{Results}
 We executed \textsc{Max-LPA} using the setup discussed above.
 Table~\ref{tab:multi} shows the number of simulations out of 50 simulations per
 $n$ and $c$ values for which it ended up in
 a single community for each pair of $n$ and $c$. If the input graph is
 disconnected then obviously there will be multiple communities.  Therefore, we
 also noted
 number of simulations for which the graph was connected 
 and this number is shown in the brackets.
 \begin{table}[!ht]
 \centering
 \caption{This table shows simulations on \er graphs $G(n, p)$ where
 $p = \frac{c \log n}{n}$.
 Each entry in the table shows the
 number of simulations out of 50 simulations per $n$ and $c$ values in
 which a single community is 
 declared by \textsc{Max-LPA} and number of simulations in which
 the graph $G(n,p)$ was connected is shown in brackets.\label{tab:multi}} 
 \begin{tabularx}{\textwidth}{|X|X|X|X|X|}
 \hline
   $n$ & $c=1$ & $c=1.2$ & $c=1.5$ & $c=1.7$\\ \hline
 1000 & 44 (50) & 47 (47) & 50 (50) & 50 (50)\\ \hline
 2000 & 42 (46) & 47 (50) & 47 (50) & 50 (50)\\ \hline
 4000 &  45 (47) &  49 (50) & 50 (50) & 50 (50) \\ \hline
 8000 &  47 (48)  & 50 (50) & 50 (50) & 50 (50) \\ \hline
 16000 &  49 (50) & 50 (50) & 50 (50) & 50 (50) \\ \hline
 32000 &  49 (50) & 50 (50) & 50 (50) & 50 (50) \\ \hline
 64000 &  50 (50) & 50 (50) & 50 (50) & 50 (50)\\ \hline
 128000 &  50 (50) & 50 (50) & 50 (50) & 50 (50)\\ \hline
  \end{tabularx}
 \end{table}
 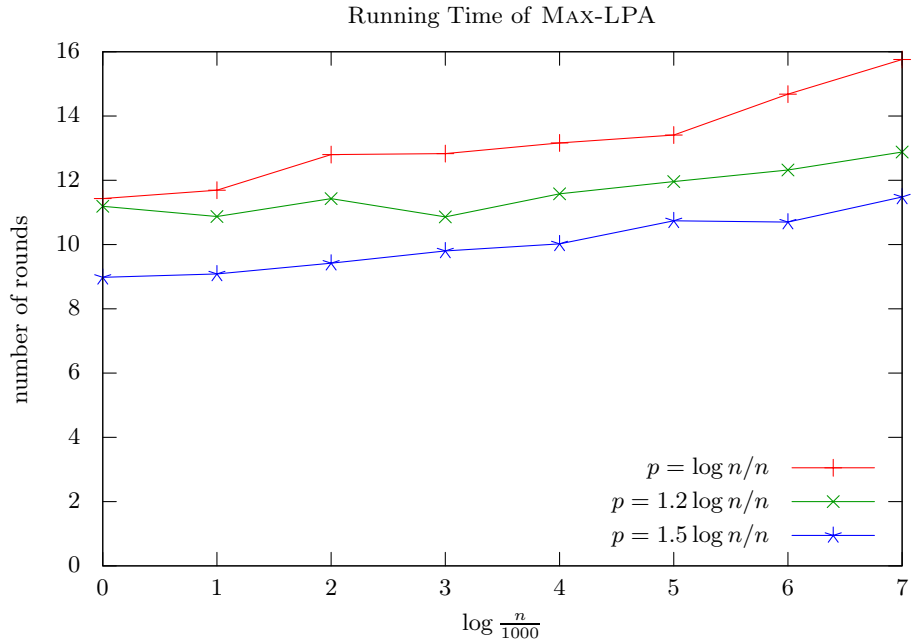
\begin{figure}[!ht]
 \centering
  \begin{tikzpicture}[gnuplot]
\gpsolidlines
\gpcolor{gp lt color border}
\gpsetlinetype{gp lt border}
\gpsetlinewidth{1.00}
\draw[gp path] (1.320,0.985)--(1.500,0.985);
\draw[gp path] (11.947,0.985)--(11.767,0.985);
\node[gp node right] at (1.136,0.985) { 0};
\draw[gp path] (1.320,1.840)--(1.500,1.840);
\draw[gp path] (11.947,1.840)--(11.767,1.840);
\node[gp node right] at (1.136,1.840) { 2};
\draw[gp path] (1.320,2.695)--(1.500,2.695);
\draw[gp path] (11.947,2.695)--(11.767,2.695);
\node[gp node right] at (1.136,2.695) { 4};
\draw[gp path] (1.320,3.550)--(1.500,3.550);
\draw[gp path] (11.947,3.550)--(11.767,3.550);
\node[gp node right] at (1.136,3.550) { 6};
\draw[gp path] (1.320,4.405)--(1.500,4.405);
\draw[gp path] (11.947,4.405)--(11.767,4.405);
\node[gp node right] at (1.136,4.405) { 8};
\draw[gp path] (1.320,5.260)--(1.500,5.260);
\draw[gp path] (11.947,5.260)--(11.767,5.260);
\node[gp node right] at (1.136,5.260) { 10};
\draw[gp path] (1.320,6.115)--(1.500,6.115);
\draw[gp path] (11.947,6.115)--(11.767,6.115);
\node[gp node right] at (1.136,6.115) { 12};
\draw[gp path] (1.320,6.970)--(1.500,6.970);
\draw[gp path] (11.947,6.970)--(11.767,6.970);
\node[gp node right] at (1.136,6.970) { 14};
\draw[gp path] (1.320,7.825)--(1.500,7.825);
\draw[gp path] (11.947,7.825)--(11.767,7.825);
\node[gp node right] at (1.136,7.825) { 16};
\draw[gp path] (1.320,0.985)--(1.320,1.165);
\draw[gp path] (1.320,7.825)--(1.320,7.645);
\node[gp node center] at (1.320,0.677) { 0};
\draw[gp path] (2.838,0.985)--(2.838,1.165);
\draw[gp path] (2.838,7.825)--(2.838,7.645);
\node[gp node center] at (2.838,0.677) { 1};
\draw[gp path] (4.356,0.985)--(4.356,1.165);
\draw[gp path] (4.356,7.825)--(4.356,7.645);
\node[gp node center] at (4.356,0.677) { 2};
\draw[gp path] (5.874,0.985)--(5.874,1.165);
\draw[gp path] (5.874,7.825)--(5.874,7.645);
\node[gp node center] at (5.874,0.677) { 3};
\draw[gp path] (7.393,0.985)--(7.393,1.165);
\draw[gp path] (7.393,7.825)--(7.393,7.645);
\node[gp node center] at (7.393,0.677) { 4};
\draw[gp path] (8.911,0.985)--(8.911,1.165);
\draw[gp path] (8.911,7.825)--(8.911,7.645);
\node[gp node center] at (8.911,0.677) { 5};
\draw[gp path] (10.429,0.985)--(10.429,1.165);
\draw[gp path] (10.429,7.825)--(10.429,7.645);
\node[gp node center] at (10.429,0.677) { 6};
\draw[gp path] (11.947,0.985)--(11.947,1.165);
\draw[gp path] (11.947,7.825)--(11.947,7.645);
\node[gp node center] at (11.947,0.677) { 7};
\draw[gp path] (1.320,7.825)--(1.320,0.985)--(11.947,0.985)--(11.947,7.825)--cycle;
\node[gp node center,rotate=-270] at (0.246,4.405) {number of rounds};
\node[gp node center] at (6.633,0.215) {$\log \frac{n}{1000}$};
\node[gp node center] at (6.633,8.287) {Running Time of \textsc{Max-LPA}};
\node[gp node right] at (10.299,2.290) {$p=\log n/n$};
\gpcolor{gp lt color 0}
\gpsetlinetype{gp lt plot 0}
\draw[gp path] (10.483,2.290)--(11.579,2.290);
\draw[gp path] (1.320,5.872)--(2.838,5.983)--(4.356,6.457)--(5.874,6.470)--(7.393,6.612)%
  --(8.911,6.717)--(10.429,7.261)--(11.947,7.722);
\gpsetpointsize{8.00}
\gppoint{gp mark 1}{(1.320,5.872)}
\gppoint{gp mark 1}{(2.838,5.983)}
\gppoint{gp mark 1}{(4.356,6.457)}
\gppoint{gp mark 1}{(5.874,6.470)}
\gppoint{gp mark 1}{(7.393,6.612)}
\gppoint{gp mark 1}{(8.911,6.717)}
\gppoint{gp mark 1}{(10.429,7.261)}
\gppoint{gp mark 1}{(11.947,7.722)}
\gppoint{gp mark 1}{(11.031,2.290)}
\gpcolor{gp lt color border}
\node[gp node right] at (10.299,1.840) {$p=1.2 \log n/n$};
\gpcolor{gp lt color 1}
\gpsetlinetype{gp lt plot 1}
\draw[gp path] (10.483,1.840)--(11.579,1.840);
\draw[gp path] (1.320,5.769)--(2.838,5.633)--(4.356,5.871)--(5.874,5.628)--(7.393,5.935)%
  --(8.911,6.098)--(10.429,6.252)--(11.947,6.491);
\gppoint{gp mark 2}{(1.320,5.769)}
\gppoint{gp mark 2}{(2.838,5.633)}
\gppoint{gp mark 2}{(4.356,5.871)}
\gppoint{gp mark 2}{(5.874,5.628)}
\gppoint{gp mark 2}{(7.393,5.935)}
\gppoint{gp mark 2}{(8.911,6.098)}
\gppoint{gp mark 2}{(10.429,6.252)}
\gppoint{gp mark 2}{(11.947,6.491)}
\gppoint{gp mark 2}{(11.031,1.840)}
\gpcolor{gp lt color border}
\node[gp node right] at (10.299,1.390) {$p=1.5 \log n/n$};
\gpcolor{gp lt color 2}
\gpsetlinetype{gp lt plot 2}
\draw[gp path] (10.483,1.390)--(11.579,1.390);
\draw[gp path] (1.320,4.824)--(2.838,4.869)--(4.356,5.012)--(5.874,5.175)--(7.393,5.269)%
  --(8.911,5.576)--(10.429,5.559)--(11.947,5.893);
\gppoint{gp mark 3}{(1.320,4.824)}
\gppoint{gp mark 3}{(2.838,4.869)}
\gppoint{gp mark 3}{(4.356,5.012)}
\gppoint{gp mark 3}{(5.874,5.175)}
\gppoint{gp mark 3}{(7.393,5.269)}
\gppoint{gp mark 3}{(8.911,5.576)}
\gppoint{gp mark 3}{(10.429,5.559)}
\gppoint{gp mark 3}{(11.947,5.893)}
\gppoint{gp mark 3}{(11.031,1.390)}
\gpcolor{gp lt color border}
\gpsetlinetype{gp lt border}
\draw[gp path] (1.320,7.825)--(1.320,0.985)--(11.947,0.985)--(11.947,7.825)--cycle;
\gpdefrectangularnode{gp plot 1}{\pgfpoint{1.320cm}{0.985cm}}{\pgfpoint{11.947cm}{7.825cm}}
\end{tikzpicture}
  \caption{Number of rounds for \textsc{Max-LPA} when executed on sparse
 \er (averaged over simulations where it ended with a single
 community out of 50 simulations per $n$ and $p$). \label{fig:rounds}}
 \end{figure}
 
 It is well known that $p=\frac{\log n}{n}$ is a
 threshold for connectivity in \er graphs and therefore we are getting few runs
 for $c=1$ where the input graph was disconnected. From Table~\ref{tab:multi},
 we can say that \textsc{Max-LPA} when executed on \er graphs
 with $p=\frac{c\log n}{n}$ and $c>1$, with high probability, terminate with one
 community. It also seem to be the case that as $c$ increases, we are getting
 more single community runs. This is because as $c$ increases, the graph become
 more dense. 
 
 Figure~\ref{fig:rounds} shows a plot of the number of rounds \textsc{Max-LPA}
takes to converge on $G(n, p)$ as $n$ increases,
 averaged over all simulations which resulted in a single community at the
 end of the execution. The running time seems to grow in a linear fashion with
 logarithm of graph size. Also as $c$ increases the running time decreases,
 which implies that as the graph becomes more dense \textsc{Max-LPA} converges
more
 quickly to a single community. Our results lead us to conjecture that when 
 \textsc{Max-LPA} is executed on \er graphs $G(n,p)$ with $p=O(\frac{\log
 n}{n})$ it will, with high probability, terminate with a single community in
$O(\log n)$
 rounds.
 
 Table~\ref{tab:clustered} shows the number of simulations out of
 50 simulations per
 $n$ and $c$ values for which \textsc{Max-LPA} correctly identified the
 partition $\Pi$ when executed on $G(\Pi,\pi,p')$ for $p'=\frac{0.6}{n}$. From
 previous results in Table~\ref{tab:multi}, for $c=1.5$ \textsc{Max-LPA}
 declared a single community when executed on $G(n,p)$ w.h.p. Therefore in this
 experiments we started with $c=1.5$. But for $c=1.5$, the influence from
 the nodes from other partition is significant. As $c$ increases this influence
 is not significant compared to the influence from nodes within the same
 partition. 
 \begin{table}[!ht]
 \centering
 \caption{This table shows simulations of \textsc{Max-LPA} on $G(\Pi, \pi, p')$
 with $\Pi = (V_1, V_2)$, $|V_1| = |V_2| = n/2$, $\pi = (p, p)$, where
 $p = \frac{c \log n}{n}$ and $p' = \frac{0.6}{n}$.
 Each entry in the table shows, for particular $n$ and $c$ values, the number of
 simulations out of 50 
 in which \textsc{Max-LPA} identified two communities $V_1$ and $V_2$.
 The number of simulations in which
 graph was connected is shown in brackets.\label{tab:clustered}} 
 \begin{tabularx}{\textwidth}{|X|X|X|X|}
 \hline
  $n$ & $c=1.5$ & $c=2$ & $c=4$ \\ \hline
 1000  & 22 (45)  & 39 (50)  & 50 (50)  \\ \hline
 2000  & 21 (39)  & 40 (50)  & 50 (50)  \\ \hline
 4000  & 22 (36)  & 47 (50)  & 50 (50)  \\ \hline
 8000  & 14 (38)  & 47 (50)  & 50 (50)  \\ \hline
 16000  & 26 (35)  & 49 (49)  & 50 (50)  \\ \hline
 32000  & 17 (33)  & 49 (49)  & 50 (50)  \\ \hline
 64000  & 26 (34)  & 46 (50)  & 50 (50)  \\ \hline
 128000  & 5 (35)  & 47 (47)  & 50 (50)  \\ \hline
 \end{tabularx}
 \end{table}

\section{Future Work}
\label{section:futureWork}  

We believe that with some refinements, the analysis technique used to show 
$O(\log n)$-rounds convergence of \textsc{Max-LPA} on paths, can be used
to show poly-logarithmic convergence on sparse graphs in general,
e.g., those with degree bounded by a constant.
This is one direction we would like to take our work in.

At this point the techniques used in Section~\ref{sec:er} do not seem applicable to more sparse
clustered \er graphs. But if we were willing to allow more time for
\textsc{Max-LPA} to converge, say
$O(\log n)$ rounds, could we significantly
weaken the requirements on the $p_i$'s and $p'$?
Specifically, could we permit an intra-community probability $p_i$ to become as
small as 
$c \log n/n$ for some constant $c > 1$?
Similarly, could we permit the inter-community probability $p'$ to come much
closer to the smallest
$p_i$, say within a constant factor?
This is another direction for our research.

\subsubsection*{Acknowledgments.}
We would like to thank James Hegeman for helpful discussions and for
some
insightful comments.

\bibliography{socialnetwork}
\end{document}